\documentclass[a4paper,USenglish,cleveref]{lipics-v2021}

\hideLIPIcs
\nolinenumbers

\usepackage[ruled]{algorithm}

\newcommand{\Id}[1]{\ensuremath{\mathit{#1}}}

\newcommand{\realrange}[2]{\left[#1, #2\right]}

\newcommand{\unitrange}[2]{\realrange{0}{1}}

\newcommand{\Oh}[1]{\mathrm{O}\!\left( #1\right)}

\newcommand{\Th}[1]{\mathrm{\Theta}\!\left( #1\right)}

\newcommand{\Om}[1]{\mathrm{\Omega}\!\left(#1\right)}

\newcommand{\llabel}[1]{\label{\labelprefix:#1}}
\newcommand{\labelprefix}{} 

\newcommand{\punkt}{\text{.}}  
\newcommand{\komma}{\text{,}} 

\newcommand{\labelcommand}{}
\newcommand{\captiontext}{}
\newsavebox{\buchalgorithmparam}
\newcounter{lineNumber}
\newenvironment{buchalgorithmpos}[3]{%
\renewcommand{\labelcommand}{#2}%
\renewcommand{\captiontext}{#3}%
\sbox{\buchalgorithmparam}{\parbox{\textwidth}{#3}}%
\begin{figure}[#1]\begin{code}\setcounter{lineNumber}{1}}
{\end{code}%
\caption{\llabel{\labelcommand}\captiontext}\end{figure}}

{\end{buchalgorithmpos}}

\newenvironment{code}{\noindent%
\begin{tabbing}%
\hspace{1.5em}\=\hspace{1.5em}\=\hspace{1.5em}\=\hspace{1.5em}\=\hspace{1.5em}\=\hspace{1.5em}\=\hspace{1.5em}\=%
\hspace{1.5em}\=\hspace{1.5em}\=\hspace{1.5em}\=\hspace{1.5em}\=\hspace{1.5em}\=%
\kill}{\end{tabbing}}

\newcommand{\Procedure}{{\bf Procedure\ }}

\newlength{\mynegthinspace}
\newlength{\mysmallspace}
\newcommand{\Repeat}   {{\bf repeat\ }}
\newcommand{\Until}    {{\bf until\ }}

\newcommand{\Is}{\ensuremath{\mathbin{:=}}}

\newcommand{\If}       {{\bf if\ }}

\newcommand{\Then}     {{\bf then\ }}

\newcommand{\Return}   {{\bf return\ }}

\newcommand{\Or}       {\ensuremath{\vee}}
\newcommand{\Rem}[1]   {{\bf /\!/\hspace{0.5mm}{\rm#1}}}

\newdimen\endofsize\endofsize=0.5em

\newcommand{\donotshow}[1]{}

\newcommand{\ignore}[1]{}

\newcommand{\mbegin}{\{\ \ }

\newlength{\mleftindent}
\newlength{\mindent}
\newlength{\mboxwidth}

\newlength{\preprogramskip}
\newlength{\postprogramskip}
\newlength{\mexpwidth}
\newlength{\mexpindent}

\newlength\fboxsepsave

\newcommand{\myurl}[1]{{\footnotesize \url{#1}}}

\usepackage{xcolor}
\usepackage{xspace}
\usepackage{booktabs}
\usepackage{siunitx}
\usepackage{mathtools}
\usepackage{csquotes}
\usepackage[section]{placeins}
\newcommand{\myunderline}[1]{{\kern-0.05em\underline{\kern0.05em #1\kern-0.05em}\kern0.05em}}
\newcommand{\myoverline}[1]{{\kern0.05em\overline{\kern-0.05em #1\kern0.05em}\kern-0.05em}}

\newcommand{\set}[1]{\left\{ #1\right\}}
\newcommand{\titleText}{Engineering MultiQueues: Fast Relaxed Concurrent Priority Queues}

\hypersetup{
  pdftitle={\titleText},
  colorlinks=true,
  pdfauthor={Marvin Williams, Peter Sanders, and Roman Dementiev},
  pdfsubject={}
}

\title{\titleText}

\author{Marvin Williams}{Karlsruhe Institute of Technology, Germany}{williams@kit.edu}{}{}
\author{Peter Sanders}{Karlsruhe Institute of Technology, Germany}{sanders@kit.edu}{}{}
\author{Roman Dementiev}{Intel Deutschland GmbH, Munich, Germany}{roman.dementiev@intel.com}{}{}
\authorrunning{M. Williams, P. Sanders and R. Dementiev}
\authorrunning{M. Williams, P. Sanders and R. Dementiev}

\Copyright{M. Williams, P. Sanders and R. Dementiev}


\relatedversion{
  This paper continues a previous paper that only appeared as a preprint  and
  SPAA brief announcement listed below.  New contributions compared to that are
  measures for better locality, more analysis of element delays, and more
  realistic experiments.
}
\relatedversiondetails[cite=RihaniSD15]{Preprint}{https://arxiv.org/abs/1411.1209}
\relatedversiondetails[cite=RSD14]{SPAA brief announcement}{https://dl.acm.org/doi/10.1145/2755573.2755616}

\supplement{The source code of our MultiQueue implementation as well as
the code for the experimental evaluation can be found as follows:}
\supplementdetails{Software}{https://github.com/marvinwilliams/multiqueue}
\supplementdetails[subcategory=experimental evaluation]{Software}{https://github.com/marvinwilliams/multiqueue_experiments}

\funding{This project has received funding from the European Research Council (ERC) under the grant \textquote{ScAlBox} (No. 882500).\\
\includegraphics[width=3cm]{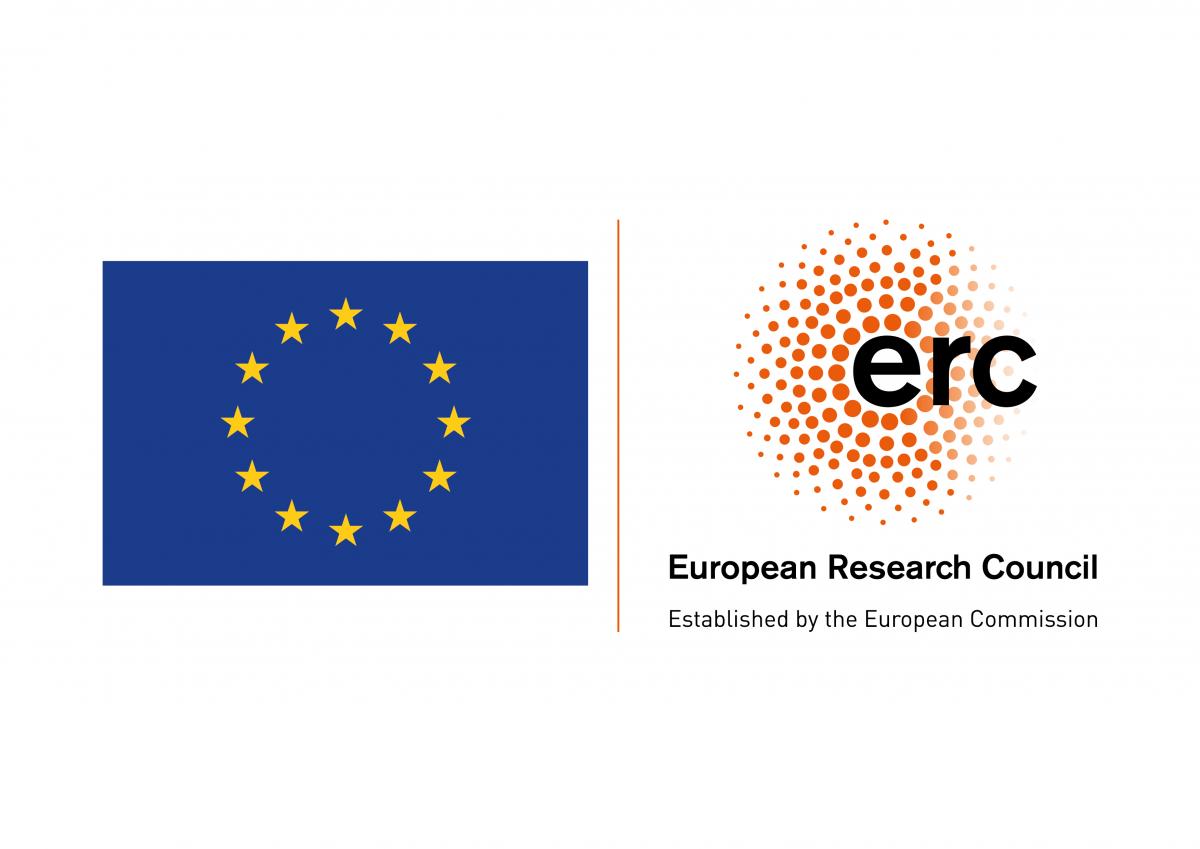}\\\textcolor{white}{.}}

\begin{document}

\maketitle

\begin{abstract}
  Priority queues with parallel access are an attractive data structure for
  applications like prioritized online scheduling, discrete event simulation,
  or greedy algorithms. However, a classical priority queue constitutes a
  severe bottleneck in this context, leading to very small throughput. Hence,
  there has been significant interest in concurrent priority queues with
  relaxed semantics.  We investigate the complementary quality
  criteria \emph{rank error} (how close are deleted elements to the global
  minimum) and \emph{delay} (for each element $x$, how many elements with lower
  priority are deleted before $x$).  In this paper, we introduce
  \emph{MultiQueues} as a natural approach to relaxed priority queues based on
  multiple sequential priority queues.  Their naturally high theoretical
  scalability is further enhanced by using three orthogonal ways of batching
  operations on the sequential queues.  Experiments indicate that MultiQueues
  present a very good performance--quality tradeoff and considerably
  outperform competing approaches in at least one of these aspects.

  We employ a seemingly paradoxical technique of \textquote{wait-free locking} that
  might be of more general interest to convert sequential data structures to
  relaxed concurrent data structures.
\end{abstract}

\clearpage
\section{Introduction}\label{intro}

Priority queues (PQs) are a fundamental data structure for many applications.
They manage a set of elements and support operations for efficiently inserting
elements and deleting the smallest element (\Id{deleteMin}).  Whenever we have
to dynamically reorder operations performed by an algorithm, PQs can turn out
to be useful. Examples include job scheduling, graph algorithms for shortest
paths and minimum spanning trees, discrete event simulation, best first
branch-and-bound, and other best first heuristics.

\begin{figure}[b]
  \centering
  \includegraphics[width=\textwidth]{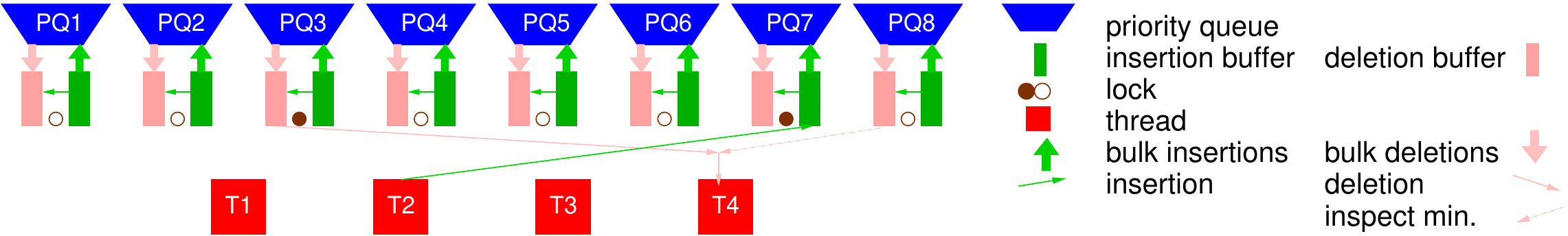}
  \caption{\label{fig:multiQueue}A MultiQueue with 8 sequential queues and 4 threads. Thread T2 currently locks PQ7 to insert an element. T4 inspected PQ3 and PQ8 and now locks PQ8 to remove its smallest element.}
\end{figure}

On modern parallel hardware, we often have the situation that $p$ parallel
threads want to access the PQ concurrently both to insert and to delete
elements. This is problematic for several reasons. First of all, even the
semantics of a parallel PQ is unclear.  The classical notion of serializability
is not only expensive to achieve but also not very useful from an application
point of view. For example, in a branch-and-bound application, a serializable
parallel PQ could arbitrarily postpone insertion and corresponding deletion of
a search tree node on the path leading to the eventual solution. This makes the
application arbitrarily slower than the sequential one.  The \emph{ideal}
semantics of a concurrent PQ (CPQ) would be that any element for which an
insertion has started becomes visible instantaneously for deletion from any
other thread. This is unattainable for fundamental physical reasons but can
serve as a basis for defining the quality of relaxed priority queues (RPQs).
In \cref{prel} we use this approach to define the complementary notions of the
\emph{rank error} of deleted elements and the \emph{delay} of elements that are
overtaken by larger elements.  After discussing related work in \cref{rel},
\cref{mq} describes the main contribution of this paper: MultiQueues are a
simple approach to CPQs based on $c\cdot p$ sequential PQs (SPQs) for some
constant $c>1$. Insertions go to a random SPQ so that each of them contains a
representative sample of the globally available elements. Deleted elements are
the smaller of the minima of \emph{two} randomly chosen SPQs. By choosing two
rather than one queue, fluctuations in the distribution of queue elements are
stabilized.  Consistency is maintained by locking queues that are changed.  See
\cref{fig:multiQueue} for an example.  Since there are more queues than
threads, no thread ever has to wait for a lock. Thus, we achieve a lock-free
(and even wait-free) algorithm despite using locks which is an interesting
feature of MultiQueues.

Although MultiQueues scale better than competitors both in theory and practice,
they have the practical problem that they lack cache locality -- in each
operation, a thread accesses several cache lines from randomly chosen SPQs
which are rarely reused but usually cause cache invalidation costs later. We
therefore introduce three orthogonal measures for improving locality: (1)
Insertion and deletion \emph{buffers} ensure that most operations access only a
small number of cache lines. (2) We use sequential queues that allow \emph{bulk
access} in a more cache-efficient way than using several single-element
operations.  (3) Threads optionally \emph{stick} to the same set of queues for
several consecutive operations.  In the analysis (\cref{ss:analysis}), we show
that MultiQueue operations are almost as fast as their sequential counterparts
-- scaling linearly with the number of threads. We are also able to analyze
rank error and delay under moderately simplifying assumptions -- it is linear in
expectation and $\Oh{p\log p}$ with high probability.  \Cref{exp} summarizes an
extensive experimental evaluation of MultiQueues including a comparison with
alternative approaches. \Cref{concl} concludes the paper.

\section{Preliminaries}\label{prel}

A priority queue $pq$ represents a set of elements. We use $n=|pq|$ for the size
of the queue. Classical priority queues support the operations \Id{insert} for
inserting an element and \Id{deleteMin} for obtaining and removing the smallest
element.  The most frequently used SPQ is the binary heap \cite{Wil64}.

A \emph{relaxed CPQ} does not require the \Id{deleteMin} operation to
return the minimum element.  The \emph{rank} of an element of a set
$M$ is its position in a sorted representation of $M$.  A natural
quality criterion for the \Id{deleteMin} operation is the \emph{rank
  error} of the returned element, i.e., the number of elements in the
CPQ at the time of deletion that are smaller than the returned
element.  Over the entire use of the CPQ, one can look at the mean
rank, the largest observed rank, or, more generally, the distribution
of observed ranks.  A complementary measure that has previously been
largely neglected is the \emph{delay} of a queue element $x$. The
delay of $x$ is the number of elements with lower priority than $x$
that are deleted before $x$.  This measure is important because it can
be closely tied to the performance of applications. For example,
consider a CPQ used in an optimization problem where the currently
best queue element leads to the ultimate solution. If the CPQ from now
on always remove the second best element, it has excellent rank error
but will never find the solution because the best element in
infinitely delayed.

It is difficult for a CPQ to consistently and efficiently maintain its exact
\emph{size}. A straightforward way to approximate the size can for example be
implemented similarly to concurrent hash tables \cite{MaierSD19}.  While
knowing the size is not of direct importance for most applications, many of
them need some kind of \emph{termination detection}.  Sequential algorithms
would often terminate after they found the queue to be \emph{empty}.  Even the
check for emptiness is difficult in concurrent settings since there is no way
to know whether some other thread has concurrently called an insertion
operation (or is about to do so). Hence, we view termination detection as a
problem that has to be solved by the application. We adopt the semantics of
most RPQs we are aware of that the \Id{deleteMin} operation is allowed to fail
-- implying that the queue could not find any remaining elements without being
able to prove that the queue is actually empty.  A failed \Id{deleteMin} has
the current size of the queue as its rank error, i.e., as if $\infty$ was
removed from the queue.

An algorithm is \emph{wait-free} if it is guaranteed to make progress
in a bounded number of steps \cite{HerSha08}. Since we discuss
randomized algorithms, we use this term also if the bound is
probabilistic, i.e., if the expected number of steps is bounded.

\section{Related Work}\label{rel}
There has been considerable work on bulk parallel priority queues (BPQs) in the 1990s \cite{DeoPra92,RanEtAl94,San98a}.
BPQs differ from CPQs in that they assume synchronized batched operation but
they are still relevant as a source of ideas for asynchronous implementations.
Indeed, Sanders \cite{San98a} already discusses how these data structures could
be made asynchronous in principle: Queue server threads periodically and
collectively extract the globally smallest elements from the queue moving them
into a buffer that can be accessed asynchronously. Note that within this
buffer, priorities can be ignored since all buffered elements have a low rank.
Similarly, an insertion buffer keeps recently inserted elements.  Moreover, the
best theoretical  results on BPQs give us an idea of how well RPQs should scale
asymptotically. For example, Sanders' BPQ \cite{San98a} removes the $p$
smallest elements of the BPQ in time $\Oh{\log n}$. This indicates that the
worst-case rank error and delay close to \emph{linear} in the number of threads
should be achievable.

Sanders' BPQ \cite{San98a} is based on the very simple idea to maintain a local
SPQ on each thread and to send inserted elements to random threads. This idea
is older, stemming from Karp and Zhang \cite{KarZha93}.  This approach could
actually be used as an RPQ. Elements are inserted into the SPQ of a randomly
chosen thread. Each thread deletes elements from its local SPQ. It is shown
that this approach leads to only a constant factor more work compared to a
sequential algorithm for a globally synchronized branch-and-bound application
where producing and consuming elements takes constant time. Unfortunately, for
a general asynchronous RPQ, the Karp-Zhang-queue \cite{KarZha93} has
limitations since slow threads could \textquote{sit} on small elements, while
fast threads would busily process elements with high rank -- in the worst case,
the rank error could grow arbitrarily large.  Our MultiQueue builds on the
Karp-Zhang-queue \cite{KarZha93}, adapting it to a shared memory setting,
decoupling the number of queues from the number of threads, and, most
importantly, using a more robust protocol for \Id{deleteMin}.

Many previous CPQs are based on the SkipList data structure \cite{Pug90}. At
its bottom, the SkipList is a sorted linked list of elements. Search is
accelerated by additional layers of linked lists.  Each list in level $i$ is a
random sample of the list in level $i-1$.  Many previous CPQs delete the exact
smallest element \cite{ShaLot00,SunTsi03,LinJon13,CMH14}. This works well if
there are not too many concurrent \Id{deleteMin} operations competing for
deleting the same elements. However, this inevitably results in heavy
contention if very high throughput is required.  The SprayList \cite{AKLS14}
reduces contention for \Id{deleteMin} by removing not the global minimum but an
element among the $\Oh{p\log^3 p}$ smallest elements.
However, for worst case inputs, insertions can still cause heavy contention.
This is a fundamental problem of any data structure that attempts to maintain a
single globally sorted sequence.  Wimmer et al. \cite{WVTCT14} describe a RPQ for task scheduling based on a
hybrid between local and global linked lists and local SPQs. Measurements in
Alistarh et al. \cite{AKLS14} indicate that this data structure does not scale
as well as SprayLists -- probably due to a frequently accessed central linked
list.  Henzinger et al. \cite{HKPSS13} give a formal specification of RPQs and
mention a SkipList based implementation without giving details.  Interestingly,
for a relaxed FIFO-queue, the same group proposes a MultiQueue-like structure
\cite{HLHPSKS13}.

The \emph{contention avoiding priority queue (CAPQ)}
\cite{sagonas2016contention,sagonas2018contention} is based on a centralized
skip-list $S$ but switches to thread-local insertion and deletion buffers when
it detects contention on $S$.  To maintain some global view, operations will
still occasionally use $S$.  This combines high accuracy in uncontended
situations with high throughput under contention.  However, the quality penalty
for switching to local buffers is fairly high.  Assuming the centralized queue
is accessed every $m \in \Th{p}$ steps (which seems necessary to avoid contention)
the paper shows that a rank error in $\Oh{p^2}$ is guaranteed for every $m$-th step.
Nothing can be guaranteed for the remaining fraction of $1-1/m$ operations.
\footnote{Such a guarantee could also be achieved with a simplistic version of
MultiQueues where each thread inserts and deletes from a fixed queue except
that every $p$ steps a thread scans all queues for the globally smallest
element.}

For large SPQs, cache-efficient data structures are useful.  Unfortunately, the
best of these (e.g., \cite{San00b}) are not useful for CPQs since
they are only efficient in an amortized sense and locking an SPQ for an
extended period of time could lead to large rank errors and delays.

\section{MultiQueues}\label{mq}

We first describe the basic MultiQueue in \cref{ss:multiQueue} and then refine
it to improve cache locality in Sections \ref{ss:buffering}--\ref{ss:sticky}.
In \Cref{ss:analysis} we provide a simplified theoretical analysis of the run
time and quality of the MultiQueue. \Cref{ss:implementation} discusses
implementation details.

\subsection{Basic MultiQueue}\label{ss:multiQueue}
\begin{figure}
  \begin{minipage}[t]{0.4\textwidth}
    \begin{code}
      \Procedure insert$(e)$\+\\
      \Repeat\+\\
        $i\Is \Id{uniformRandom}(1..cp)$\\
        try to lock $Q[i]$\-\\
      \Until lock was successful\\
      $Q[i].\Id{insertToSPQ}(e)$\\
      unlock $Q[i]$\\
    \end{code}
  \end{minipage}\hspace*{0.1\textwidth}
  \begin{minipage}[t]{0.5\textwidth}
    \begin{code}
      \Procedure deleteMin\+\\
      \Repeat\+\\
        $i\Is \Id{uniformRandom}(1..cp)$\\
        $j\Is \Id{uniformRandom}(1..cp)$\\
        \If $Q[i].\Id{min}>Q[j].\Id{min}$ \Then swap $i$, $j$\\
        try to lock $Q[i]$\-\\
      \Until lock was successful\\
      $e\Is Q[i].\Id{deleteMinFromSPQ}$;\quad
      unlock $Q[i]$\\
      \Return $e$
    \end{code}
  \end{minipage}
\caption{\label{alg:deleteMinMQ}\label{alg:insertMQ}Pseudocode for basic
MultiQueue \Id{insert} and \Id{deleteMin}. This code assumes that an empty SPQ
returns $\infty$ as the minimum element. A return value of $\infty$ from
\Id{deleteMin} then indicates that no element could be found (which does not
guarantee that all SPQs are empty). A practical implementation might make
additional efforts to find elements when encountering empty SPQs.}
\end{figure}

The basic MultiQueue data structure is an array {\tt Q} of $c \cdot p$ SPQs
where $c$ is a tuning parameter and $p$ is the number of parallel threads.
\Cref{fig:multiQueue} gives an example with $c=2$.  Access to each local queue
is protected by a lock flag.  The \Id{insert} operation locks a random unlocked
queue $\mathtt{Q}[i]$ and inserts the element into $\mathtt{Q}[i]$.
\Cref{alg:insertMQ} gives pseudocode.  Note that this operation is wait-free
since we never wait for a locked queue.  Since at most $p$ queues can be locked
at any time, for $c>1$ we will have a constant success probability.  Hence, the
expected time for acquiring a queue is constant.  Together with the time for
insertion we get expected insertion time $\Oh{\log n}$.

An analogous implementation of \Id{deleteMin} would lock a random unlocked
queue and return its minimal element.  However, the quality of this approach
leaves a lot to be desired.  In particular, it deteriorates not only with $p$
but also with the queue size. One can show that the rank error grows
proportional to $\sqrt{n}$ due to random fluctuations in the number of
operations addressing the individual queues. Therefore we invest slightly more
effort into the \Id{deleteMin} operation by looking at \emph{two} random queues
and deleting from the one with the smaller minimum.
\Cref{alg:deleteMinMQ} gives pseudocode for the \Id{insert} and \Id{deleteMin}
operations.  Our intuition why considering two choices may be useful stems from
previous work on randomized load balancing, where it is known that placing a
ball on the least loaded of two randomly chosen machines gives a maximum load
that is very close to the average load independent of the number of allocated
balls \cite{BCSV00}.

Even when the queue is small, cache efficiency is a major issue for MultiQueues
since accessing a queue $\mathtt{Q}[i]$ from a thread $j$ will move the cache
lines accessed by the operation into the cache of thread $j$. But most likely,
some random other thread $j'$ will next need that data causing not only cache
misses for $j'$ but also invalidation traffic for $j$.

\subsection{Buffering}\label{ss:buffering}

\begin{figure}
  \begin{minipage}[t]{0.48\textwidth}
    \begin{code}
      \Procedure insertToSPQ$(e)$\+\\
        \If $D = \emptyset \Or e<\max D$ \Then\+\\
          \If $|D|<b$ \Then $D\Is D\cup \set{e}$;\quad \Return\\
          $(e,D)\Is (\max D, (D\setminus\set{\max D})\cup\set{e})$\-\\
        \If $|I|=b$ \Then flush $I$ to $M$\\
        $I\Is I\cup\set{e}$
    \end{code}
  \end{minipage}\hfill
  \begin{minipage}[t]{0.48\textwidth}
    \begin{code}
      \Procedure deleteFromSPQ\+\\
      \If $D=\emptyset$ \Then \Return $\infty$ \hspace{0.5mm}\Rem fail\\  
        $e\Is D.\Id{deleteMin}$\\
        \If $D=\emptyset$ \Then\+\\
          refill $D$ from $M\cup I$\-\\
        \Return $e$
    \end{code}
    \end{minipage}
    \caption{\label{alg:bufferPeter}Pseudocode for inserting and
      deleting elements from an already locked sequential queue
      represented by the main queue $M$, an insertion buffer $I$ and a
      deletion buffer $D$.}
\end{figure}

To reduce the average number of cache lines accessed, we represent each SPQ by
the main queue $M$, an \emph{insertion buffer} $I$ and a \emph{deletion buffer}
$D$ that is organized as a sorted ring-buffer.  Each buffer has a fixed
capacity $b$.  We maintain the invariant that (unless the SPQ is empty) $D$
contains the smallest elements of $M\cup D\cup I$.  This implies that {\bf min}
of the SPQ is always the first element of $D$.

To maintain this invariant, {\bf inserting} an element $e$ into an SPQ first
checks whether $D$ is empty or contains a larger element.  If so, $e$ is
inserted into $D$. If $D$ was not full, this finishes the insertion.
Otherwise, $e$ becomes the largest element in $D$.  If $I$ is full, it is
flushed into the main queue. Finally, $e$ is inserted into $I$.  See
\cref{alg:bufferPeter} for high-level pseudocode.

{\bf DeleteMin} is straightforward.  When $D$ is not empty, its smallest
element is removed and returned.  If this empties $D$, it is refilled from
$M\cup I$.  A simple way to do this is to first flush $I$ into $M$ and then
extract the smallest $b$ elements from $M$ into $D$.  Alternatively,
we can first refill $D$ from $M$ only and then scan through $I$ to swap
elements smaller than $\max D$ with the largest elements from $D$ in a fashion
analogous to what is done in \Id{insertSPQ}.  This has the advantage that all
interactions between the buffers and $M$ is in the form of batches of
predictable size. This will be exploited in \cref{ss:batching}.

Buffering implies that the operations most often only access the buffers themselves
and the lock (accessing and modifying the buffers requires the SPQ to be
locked).  More specifically, an insertion \textquote{usually} reads the largest
element from $D$ and then operates on $I$ (assuming that inserted elements
rarely go to $D$). \Id{deleteMin} \textquote{usually} only accesses $D$.  When
buffers are flushed or refilled, a single thread performs a batch of operations
on a single queue and thus can exploit whatever locality the main queue
supports. In binary heaps for example, insertions exhibit high locality.
Deletions at least exhibit some locality near the root, at the variable
specifying the size of the queue, and at the rightmost end of the bottom layer
of the tree.

\subsection{Batching}\label{ss:batching}

To fully exploit that MultiQueues with buffering access the main queues in a
bulk fashion we can use SPQs that directly support batch operations.  In our
prototype, we considered \emph{merging binary heaps}.  These are structured
like binary heaps but each node contains $k$ sorted elements.  The heap
invariant remains that nodes contain elements that are no smaller than the
elements in the parent node. Insertion and deletion are also analogous to
ordinary binary heaps except that compare-and-swap operations are generalized
to merge-and-split.

On the one hand, merging binary heaps yield higher cache locality for
bulk operations than binary heaps since all elements in each tree node can be
stored consecutively and fewer nodes have to be accessed. On the other hand,
they come with additional algorithmic complexity and higher worst-case access
time.  They are also not easily combined with $a$-ary heaps that are a simpler
alternative to achieve somewhat higher locality than in binary heaps.

\subsection{Stickiness}\label{ss:sticky}
As a third measure to improve cache locality we introduce the concept
of \emph{stickiness}.  The stickiness parameter $s$ controls for how
many consecutive operations a thread $t$ reuses a particular local
queue for its \Id{insert} and \Id{deleteMin} operations. After
the stickiness period of one local queue ends for $t$ or if locking
the queue fails, $t$ randomly chooses another queue to stick to.  The
intuition behind this protocol is that for large enough $s$ and $c\geq
3$, the system converges to a state where threads use disjoint queues
most of the time.
Stickiness provides a simple mechanism to trade-off increased cache
efficiency at the cost of potentially higher rank errors and increased delay.

\subsection{Analysis}\label{ss:analysis}
In this section we analyze the theoretical performance of the MultiQueue under simplifying assumptions.

\subsubsection{Running Time}
We analyze the asymptotic running time of the operations \Id{insert} and
\Id{deleteMin} of the MultiQueue in a realistic asynchronous model of
shared memory computing where $k$ threads contending to write the same machine word
need time $\Oh{k}$ to perform those operations (e.g., the aCRQW
model \cite[Section~2.4.1]{SMDD19}).

\begin{theorem}
With $c>1$, the expected execution time of the operations \Id{insert} and
\Id{deleteMin} is $\Oh{1}$ plus the time for the sequential queue access.
\end{theorem}
\begin{proof}
Whenever a thread attempts to lock a queue $q$, there are at most $p-1$ locked queues.
Hence the success probability is at least $$s\Is 1-\frac{p-1}{cp}\geq 1-\frac{1}{c} \in \Om{1}\punkt$$
Hence the expected number of attempts is
$$\sum_{i=1}^{\infty}is(1-s)^{i-1}=\frac{1}{s} \leq \frac{1}{1-\frac{1}{c}}=\frac{c}{c-1} \in \Oh{1}\punkt$$

The stickiness does not affect the expected number of attempts since it only
dictates the first queue to attempt to lock but not subsequent attempts in case
of failure.  The buffers have constant size, so all operations on them are in
$\Oh{1}$. The bulk-inserting of the insertion buffer into the queue and
refilling the deletion buffer amortize to the same asymptotic cost as accessing
a sequential queue for each element individually.
\end{proof}

Note that the above analysis even holds when threads are blocked: In the worst
case, each thread holds a lock. The overhead for another thread to avoid these
locks is already accounted for in the above proof. Hence all other threads are
guaranteed to make progress (in a probabilistic sense).  Thus MultiQueues are
probabilistically wait-free.

Using different sequential queues, we get the following bounds for the
comparison based model and for integer keys:
\begin{corollary}
MultiQueues with binary heaps need constant average insertion time and
expected time $\Oh{\log n}$ per operation for worst case operations
sequences.  For integer keys in the range $0..U$, using van Emde Boas
search trees \cite{Emde77,DKMS04} as local queues, time $\Oh{\log\log
  U}$ per operation is sufficient.
\end{corollary}

\subsubsection{Quality}

Quality analysis is still a partially open problem.  However,
we will explain how \emph{rank errors} and \emph{delays} can be estimated under
simplified but intuitive assumptions in the absence of stickiness.

Let us assume that all $m$ current elements have been allocated uniformly at
random to the local queues.  This assumption holds if there have been no
\Id{deleteMin} operations so far and no locking attempt during \Id{insert}
failed. It is an open question whether lock contention and \Id{deleteMin}
operations invalidate this assumption in practice. Furthermore, let us assume
that we choose the queues to look at for deletion randomly and no queue is
currently locked.

\paragraph*{Rank Errors}
With these assumptions, the probability to delete an element $e$ of rank $i$ is
\[ \mathsf{P}(\mathrm{rank}=i) =\left(1-\frac{2}{cp}\right)^{i-1}\frac{2}{cp} \]
The first factor expresses that the $i-1$ elements with smaller ranks are not
present at the two chosen queues. The second factor is the probability that the
element with rank $i$ \emph{is present}.  Therefore, the expected rank error in
the described situation is
\begin{equation}\label{eq:eRank}
  \sum_{i=1}^mi\mathsf{P}(\mathrm{rank}=i)\leq
  \sum_{i\geq 1}i\left(1-\frac{2}{cp}\right)^{i-1}\frac{2}{cp}=\frac{c}{2}p \in \Oh{p}\punkt
\end{equation}
The cumulative probability that the rank of $e$
is larger than $k$ is
\begin{equation}\label{eq:tail}
  \mathsf{P}(\mathrm{rank}>k)=\left(1-\frac{2}{cp}\right)^k\komma
\end{equation}
as none of the elements with a rank smaller than $k$ must be present on the two
chosen queues.
$\mathsf{P}(\mathrm{rank}>k)$ drops to $p^{-a}$ for $k=\frac{ca}{2}p\ln p$,
i.e., with probability polynomially large in $p$, we have a rank error in
$\Oh{p\log p}$.


We can also give qualitative arguments on how the performed operations change
the distribution of the elements. Insertions are random and hence move the
system towards a random distribution of elements. Deletions tend to remove more
elements from queues with small keys than from queues with large keys, thus
stabilizing the system.  Alistarh et al.  \cite{alistarh2017power,ABKLN18}
confirm our conjecture for a slightly simplified process and also show that it
suffices to only \textquote{sometimes} look at more than one key.  On the other
hand, they show that the rank error grows in an unbounded fashion if always
only a single queue is considered for deletion.

\paragraph*{Delay}
The delay of an element $e$ with rank $i$ is equal to the number of deletions of
elements with rank higher than $i$. Let $\mathsf{D}$ denote the event that a
particular \Id{deleteMin} operation delays or removes element $e$.  Let
$\mathsf{R}$ denote the event that $e$ is removed.
We now compute the conditional probability that $e$ is deleted given that
it is delayed or deleted by the \Id{deleteMin} operation.
By the definition of
conditional probability, we have
$\mathsf{P}\,(\mathsf{R} \mid \mathsf{D}) = \frac{\mathsf{P}(\mathsf{R} \cap \mathsf{D})}{\mathsf{P}(D)}$.

With $\mathsf{P}(\mathsf{R} \cap \mathsf{D}) = \mathsf{P}(\mathrm{rank} = i) =  \left(1-\frac{2}{cp}\right)^{i-1} \frac{2}{cp}$ 
and 
$\mathsf{P}(\mathsf{D}) = \mathsf{P}(\mathrm{rank} \geq i) = \left(1-\frac{2}{cp}\right)^{i-1}$,
we get $\mathsf{P}(\mathsf{R} \mid
\mathsf{D}) = \frac{2}{cp}$.
Since $\mathsf{P}(\mathsf{R} \mid \mathsf{D})$ is constant, the delay follows a
geometric distribution with an expected value of $\frac{cp}{2}$ and
values in $\Oh{p\log p}$ with high probability.

\subsection{Implementation Details}\label{ss:implementation}
\begin{figure}
  \centering
  \includegraphics[scale=0.79]{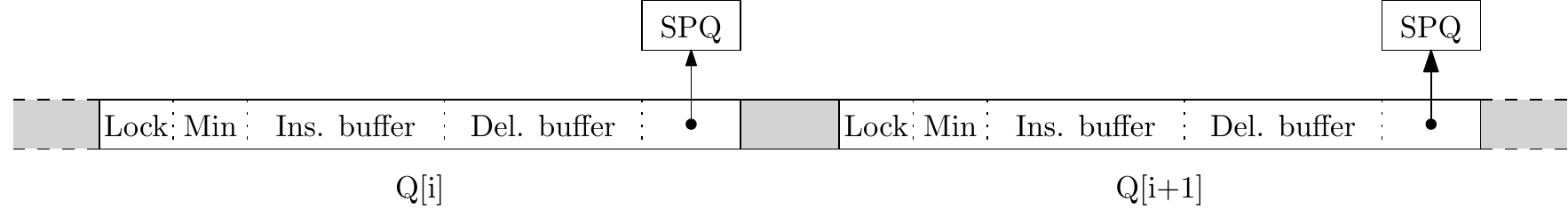}
  \caption{\label{fig:array} The array containing the lock, the
    buffers as well the pointer to the SPQ in each entry. The padding blocks
    (grey) prevent false sharing of neighboring entries. }
\end{figure}
We first give a detailed description and reasoning of the data structures
used in our implementation of the MultiQueue as described above.  We then show
implementations of the \Id{insert} and \Id{deleteMin} operations.

At its core, a MultiQueue is an array $\mathtt{Q}$ with one entry for each
local queue.  Each entry contains a lock, the insertion and deletion buffers,
and a pointer to the main queue itself. Moreover, the key of the minimum element in
the deletion buffer is redundantly stored next to the lock. Comparing the
minima of two local queues thus only involves atomic access to this key and
does not access the deletion buffer, which would require locking. Note
that this technique allows for slight inaccuracies, as the minimum in the
deletion buffer could be different from the explicitly stored minimum key if
another thread deletes the minimum from the deletion buffer during comparison.
However, consistency is not impacted by this.  The entries are padded and
aligned to \emph{cache lines} to prevent false sharing of neighboring entries.
On systems with nonuniform memory access (NUMA) the padding is extended to
\emph{virtual memory pages} that are distributed round-robin over the
NUMA-nodes. This balances memory traffic over the nodes.  See
Figure~\ref{fig:array} for an illustration.  Additionally, each thread stores
local data such as stickiness counters. False sharing between these data
structures is also avoided.


The insertion buffer is implemented as  a fixed array of size $b$, preceded by
a size counter. The deletion buffer must support efficient removal of the first
element, lookup of the last element, and insertions at arbitrary positions.  We
chose a ring buffer that stores its size and the index of the first element
upfront as the underlying data structure.  Ring buffers support random lookup and
removing the first element in constant time, while inserts at arbitrary
positions take at most $b/2$ element moves. We implement both $8$-ary heaps and
$k$-merging binary heaps as SPQs, where $k$ is a tuning parameter.  $8$-ary
heaps provide the same theoretical guarantees as binary heaps but better cache
locality.  To achieve stable worst-case access times, our implementation
allocates enough memory so that array resizes by the local queues are unlikely.
Having stable worst-case access times is relevant for MultiQueues since an
operation that exceptionally takes very long would lock a queue for a long
time.  If this queue also contains elements of low rank, their delay as well as
the rank error of elements deleted in the meantime could become large. On
machines with NUMA, the memory for each SPQ is allocated on the same NUMA node.

While our MultiQueue implementation can handle arbitrary element types, the
implementation is designed and optimized with elements of small size in mind.

\section{Experiments}\label{exp}

We first compare our theoretical analysis of rank errors and delays
with experimental results. We then perform parameter tuning and examine
different implementation variants of the MultiQueue. Afterward, we compare the
MultiQueue with other concurrent priority queues in terms of quality,
throughput, and scalability.  For these experiments, each thread repeatedly
either deletes an element from the queue or inserts a new one. This benchmark
provides good insights on the maximum throughput and quality of the queues
under high contention.

Queue elements are key-value pairs consisting of two
$32$-bit unsigned integers, where the key determines the element's priority.
The queue is filled with $n_0$ elements prior to the measurement. Unless
otherwise noted, $n_0 = 10^6$, the inserted elements are uniformly
distributed in $[0,\ldots, 2^{32} - 1]$, and the rank errors and delays are
measured over the course of $10^7$ \Id{deleteMin} operations.  To conclude this
chapter, we perform a parallel single-source shortest-path (SSSP) benchmark to
shed light on the performance under more realistic settings.

Experiments were conducted on two machines: Machine $A$ utilizes an AMD EPYC™
7702P 64-core processor. Each core runs at \SI{2.0}{\GHz} and supports two
hardware threads. The system runs on Ubuntu 20.04 with Linux kernel version
5.4.0.  Machine $B$ is a dual socket system with an Intel® Xeon® Platinum 8368
Processor with \SI{2.4}{\GHz} on each socket, yielding $2 \times 38$ cores and
$152$ available hardware threads. The system runs on SUSE Linux Enterprise
Server 15 with Linux kernel version 5.3.18.  The experiments in
Section~\ref{sec:param} were performed on machine $A$, for the comparison
experiments in Section~\ref{sec:compare} we additionally used machine $B$.  The
implementation is written in
{C\nolinebreak[4]\hspace{-.05em}\raisebox{.4ex}{\tiny\bf ++}17} and compiled
with GCC 10.2.0 with optimization level \texttt{-O3}. We use pthreads for
thread management and synchronization. Each thread is pinned to a hardware
thread.

\subsection{Measuring Rank Error and Delay of Relaxed PQs}\label{sec:errors}
Measuring the rank errors and delays in real-time imposes the practical problem
that we need to know which elements are in the queue at the time of each
\Id{deleteMin} operation. We approach this problem as follows. Each thread logs
its operations into a preallocated local vector.  The log entries get
timestamps obtained using a low-overhead high-resolution clock (we currently
use the Posix \texttt{CLOCK\_REALTIME}).
For the evaluation, we merge these logs to one global sequence $S$ of
operations (sorted by time-stamp).  $S$ is then sequentially replayed using an
augmented {B\nolinebreak[4]\hspace{-.05em}\raisebox{.4ex}{\tiny\bf +}} tree
(based on a \texttt{tlx::btree\_map} from the \texttt{tlx} library~\cite{TLX})
whose leaves are the current queue elements (sorted by key).  By augmenting
nodes with their size, this allows determining ranks of deleted elements (and
thus rank errors) in logarithmic time (e.g., \cite[Section~7.5]{SMDD19}).  We
further augment the interior nodes with delay counters $d_v$ and maintain the
invariant that the delay of a leaf $e$ is the sum of the delay counters on the
root--$e$ path. When inserting an element, we can establish the invariant by
setting $d_e$ to $-\sum_vd_v$ for $v$ on the root--$e$ path. When performing
balancing operations on the tree, we can maintain the invariant by pushing the
delay counters of the manipulated nodes downward to unchanged subtrees.

\subsection{MultiQueue Parameter Tuning}\label{sec:param}
The baseline MultiQueue uses $8$-ary heaps as SDQs and does not use stickiness
or buffers.
\begin{figure}
  \centering
  \includegraphics[width=\textwidth]{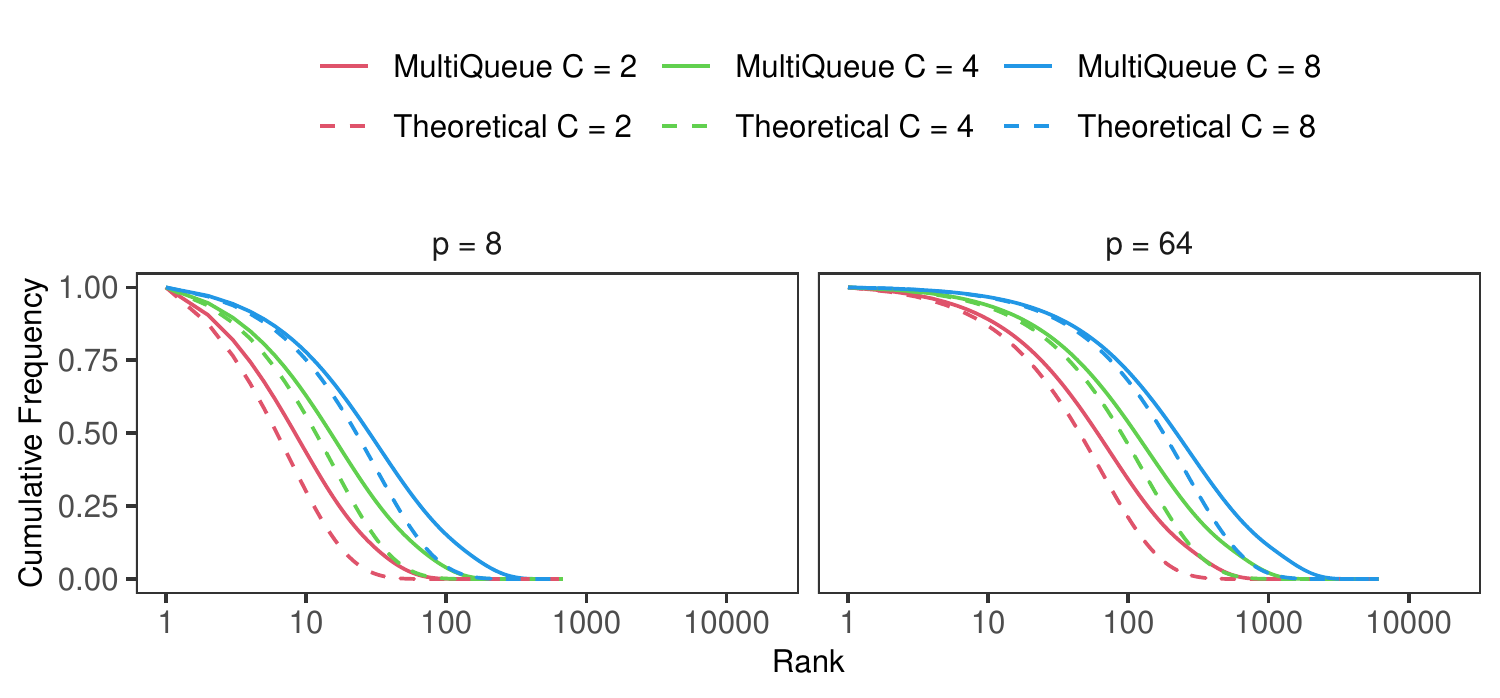}
  \caption{\label{fig:rank_theory}The rank error distribution compared to the
  distribution from the theoretical analysis for $p=8$ and $p=64$, respectively.}
\end{figure}
As seen in Figure~\ref{fig:rank_theory}, the measured rank errors follow the
predicted distribution closely independent of the number of threads and values
for $c$.  However, we cannot quite match the prediction in practice.%
\footnote{The discrepancy is not due to locking since it persists when threads
wait for a long time after each operation. Other sources could be noise in time
measurements or inaccuracies due to our random data distribution assumption.}
We now inspect the impact of various parameter configurations for the
MultiQueue on its performance.  In particular, we optimize the buffer sizes and
explore different values for $c$ and the stickiness.  The throughputs are
reported for a running time of \SI{3}{\s} and averaged over $5$ runs.

\subparagraph*{Buffer Size.}
\begin{figure}
  \begin{minipage}[t]{0.5\textwidth}
  \centering
  \includegraphics[width=\textwidth]{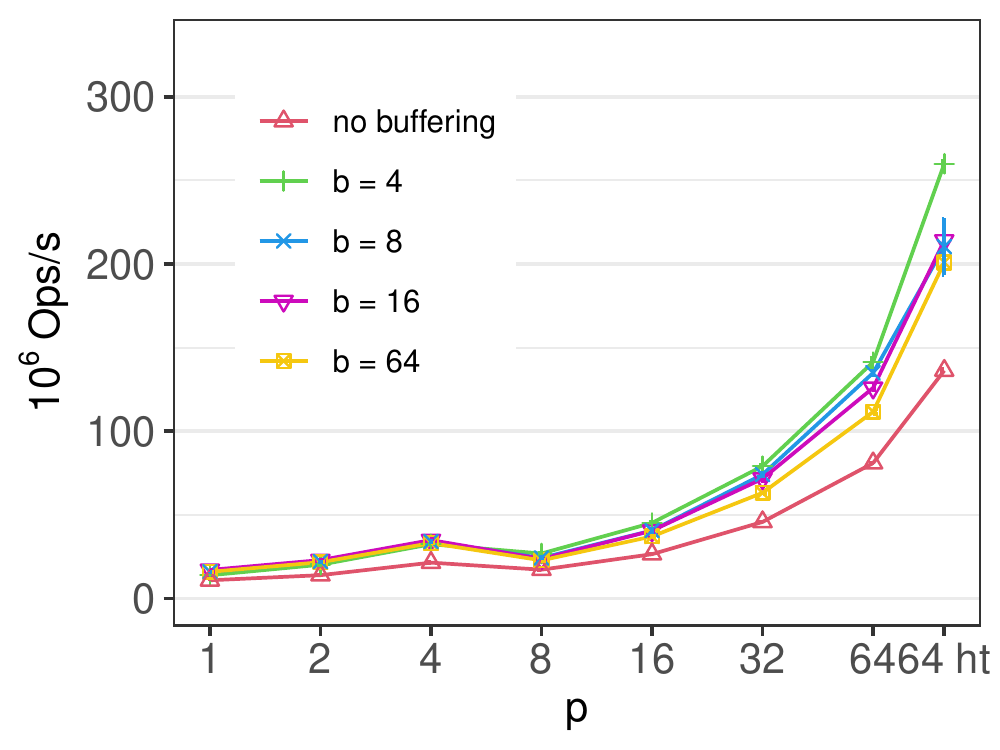}
\end{minipage}
  \begin{minipage}[t]{0.5\textwidth}
  \centering
  \includegraphics[width=\textwidth]{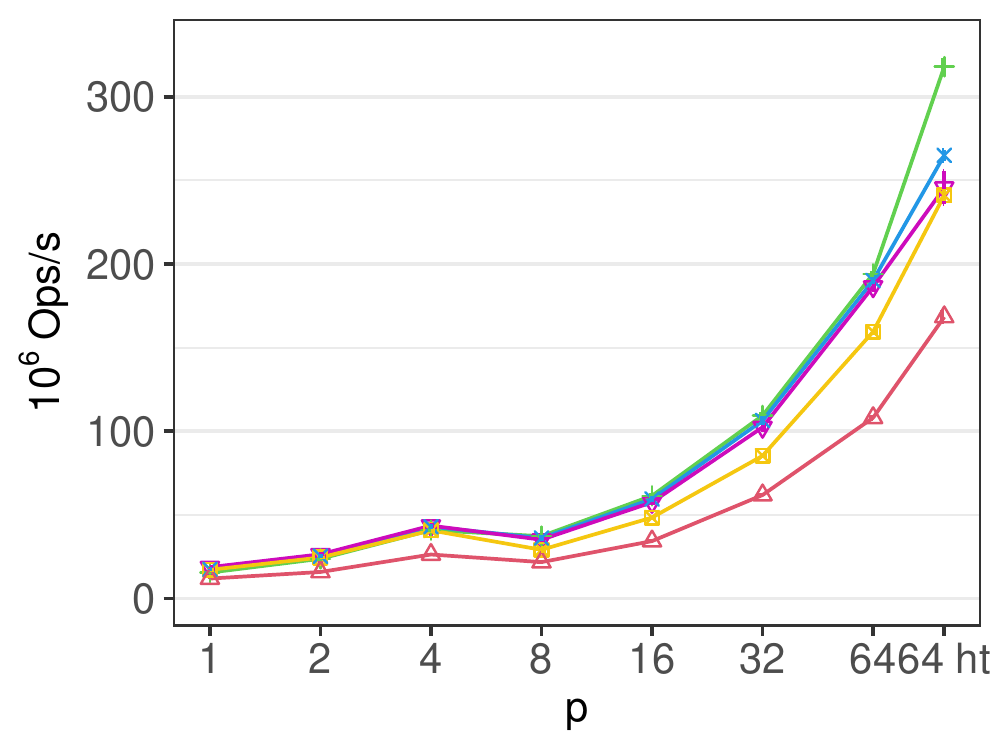}
\end{minipage}
  \caption{\label{fig:throughput_buffer_size}The impact of different buffer
  sizes to the throughput. The left plot is for $c=2$, the right one for $c=4$.
The suffix "ht" denotes active hyper-threading.}
\end{figure}
As Figure~\ref{fig:throughput_buffer_size} shows, buffering increases the
throughput of the MultiQueue considerably. However, the buffer size has to be
chosen carefully, as too large buffers hamper the performance. This is due to
the overhead of inserting elements into the deletion buffer. The difference of
buffer sizes $4$, $8$, and $16$ is negligible until hyper-threading is active, where a
buffer size of $4$ gains most. However, the larger the buffers, the less
frequent we have to access the SPQ to perform bulk operations. Accessing the
SPQ likely generates cache misses, which are especially expensive with NUMA if
the SPQ is located on another NUMA node. We therefore use a buffer size of
$b =16$ for further experiments. We have observed no significant differences in
rank errors and delays with or without buffering.

\subparagraph*{Stickiness and Number of Queues.}
\begin{table}
\centering
\begin{tabular}{rrrrr}
  \hline
  c & s & Throughput (\si{\textrm{MOps}\per\second}) & Avg. rank error & Avg. delay \\ 
  \hline
  \textbf{2} & \textbf{1} & 125.9 & (64) 103.1 & 103.1 \\ 
   & 4 & 215.5 & 413.2 & 412.5 \\ 
   & 8 & 271.6 & 900.2 & 896.0 \\ 
   & 16 & 322.5 & 1853.8 & 1827.6 \\ 
   & 64 & 411.9 & 7443.5 & 7183.2 \\ 
  \hline
  \textbf{4} & \textbf{1} & 186.5 & (128) 202.6 & 202.6 \\ 
             & \textbf{4} & 327.3 & 814.2 & 812.6 \\ 
   & 8 & 405.3 & 1685.4 & 1678.0 \\ 
   & 16 & 439.6 & 3334.9 & 3303.4 \\ 
   & 64 & 535.3 & 13349.6 & 12870.9 \\ 
   \hline
  8 & 1 & 222.5 & (256) 405.3 & 405.1 \\ 
   & 4 & 408.8 & 1598.6 & 1592.2 \\ 
   & \textbf{8} & 488.0 & 3284.9 & 3259.5 \\ 
   & 16 & 497.6 & 6608.4 & 6511.9 \\ 
   & 64 & 579.8 & 26207.6 & 24902.0 \\ 
   \hline
  16 & 1 & 244.1 & (512) 811.7 & 810.9 \\ 
  & 4 & 461.9 & 3216.8 & 3196.1 \\ 
  & \textbf{8} & 542.0 & 6415.2 & 6338.3 \\ 
  & 16 & 529.3 & 12814.4 & 12547.8 \\ 
  & 64 & 603.3 & 50854.5 & 47695.3 \\ 
   \hline
\end{tabular}

\caption{\label{fig:c_and_s} The throughput, average rank error and delay for
$p = 64$. The number in parentheses indicates the expected rank error according
to the theoretical analysis.}
\end{table}
\Cref{fig:c_and_s} shows measurements that suggest that rank errors and delays
are not just linear in the number of queues $cp$ but also in the stickiness $s$.
While rank errors and delays of MultiQueues are always very similar, this is
not generally true.  We use the configurations
$(c,s)\in\set{(4,1),(4,4),(8,8),(16,8)}$ for further benchmarks, as they
provide interesting trade-offs between throughput and quality.

\subparagraph*{$k$-Merging Heap.}
\begin{figure}[h]
  \centering
  \includegraphics[scale=0.7]{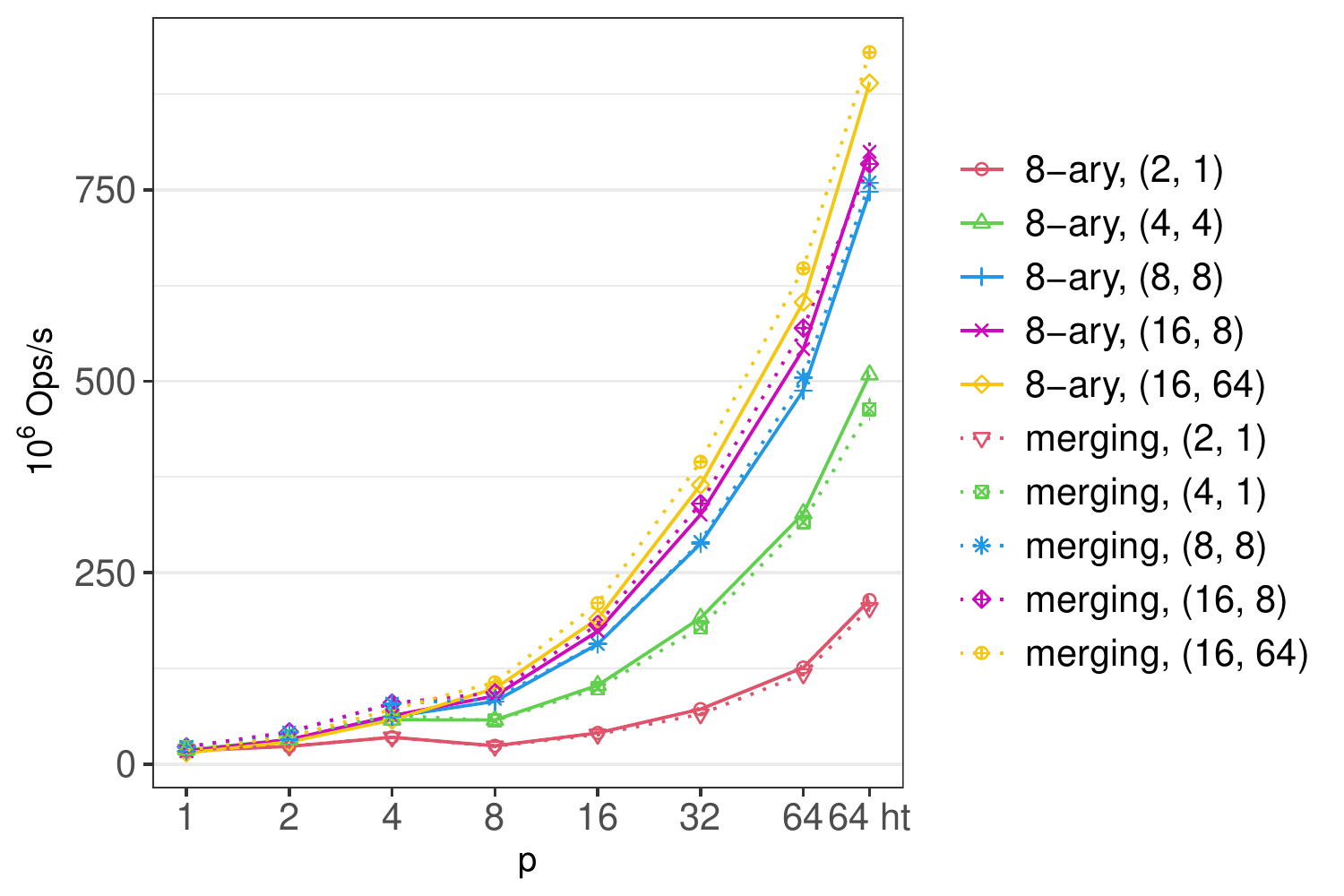}
  \caption{\label{fig:throughput_merging}Throughput comparison using $8$-ary
  heaps and $k$-merging heaps with $k = 16$. The numbers in parentheses are
$(c,s)$-pairs. The suffix "ht" denotes active hyper-threading.}
\end{figure}
Our experiments with $k$-binary merging heaps instead of $8$-ary heaps for the
SPQs indicated that merging heaps can improve the throughput for high
values of $s$ compared to $8$-ary heaps. However, the impact is moderate (see
\cref{fig:throughput_merging}) in most cases and we stick to $8$-ary heaps.
Our experiments further showed that using merging heaps has no impact on the
quality of the priority queue.

\subsection{Comparison with Other Approaches}\label{sec:compare}
\begin{figure}[b]
  \begin{minipage}[t]{0.46\textwidth}
    \centering
    \includegraphics[scale=0.65]{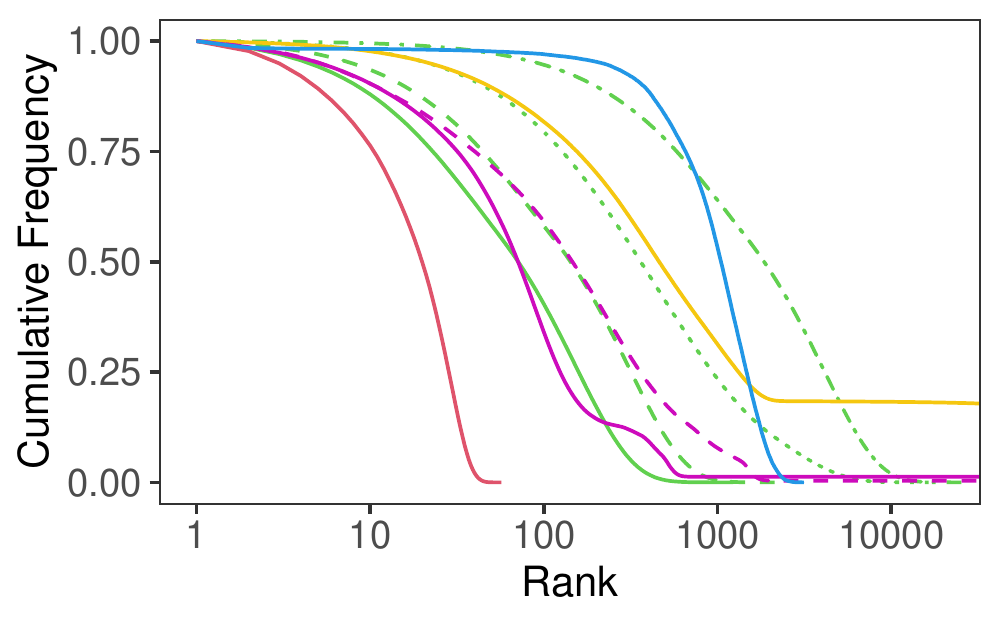}
  \end{minipage}
  \begin{minipage}[t]{0.54\textwidth}
    \centering
    \includegraphics[scale=0.65]{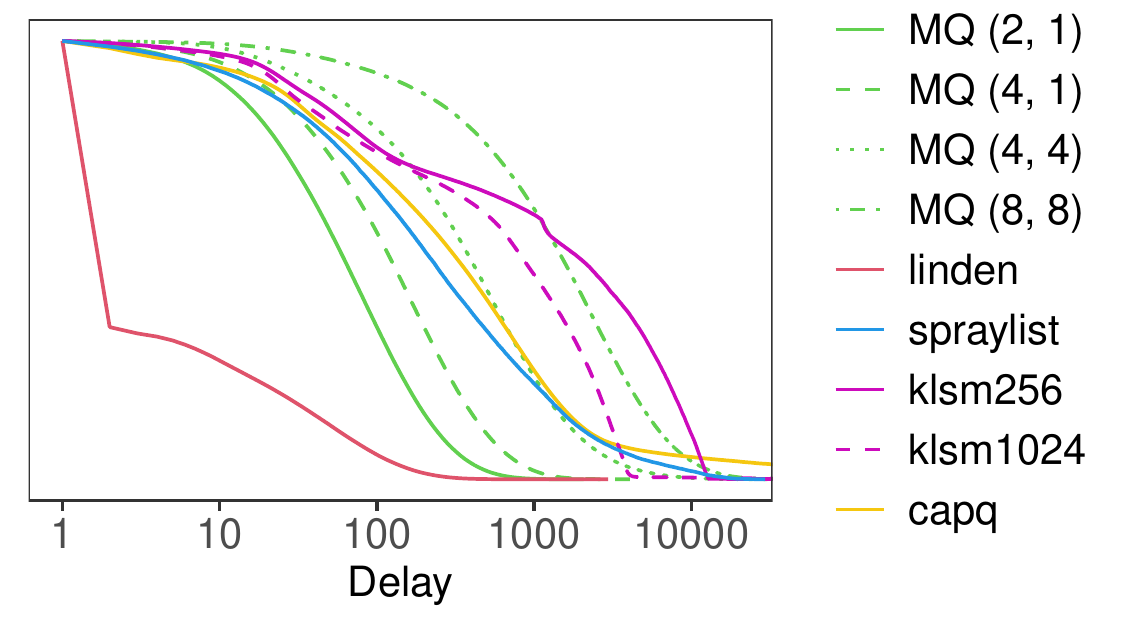}
  \end{minipage}
  \caption{\label{fig:quality_compare}Comparison of rank error and
  delay distribution for different priority queues for $10^6$ \Id{deleteMin} operations, $p=64$.}
\end{figure}
We compare the MultiQueue to the following state-of-the-art concurrent priority queues.
\begin{description}
  \item [Linden~\cite{LinJon13}] is a priority queue based on skip lists by Lindén and Jonsson. It only returns suboptimal elements in case of contention on the list head.
  \item [Spraylist~\cite{AKLS14}] relaxes skip lists by deleting elements close to the list head to reduce contention.
  \item [$k$-LSM~\cite{wimmer2015lock}] combines thread-local priority queues with a relaxed shared priority queue component.
    We use the variants with $k=256$ (\textbf{klsm256}) and $k=1024$ (\textbf{klsm1024}).
  \item [CAPQ~\cite{sagonas2016contention}] dynamically detects contention to
    switch from using a shared skip list based priority queue to thread-local buffers.
\end{description}
We used the implementation found in the Github repository of the
$k$-LSM\footnote{\url{https://github.com/klsmpq/klsm}} for all of the above
priority queues.

Figure~\ref{fig:quality_compare} shows the rank error distributions and delays
for $64$ threads. Unsurprisingly, the Linden queue yields by far the best
quality.  The $k$-LSM achieves low rank error but high delay.  For the
Spraylist it is the other way around as it has high rank error but low delay.
Multiqueues have similar delays and rank errors. With small values for $c$ and
$s$ they are more accurate than all the competitors except for the Linden
queue.  The CAPQ exhibits quality comparable  to the MultiQueue with the rather
loose setting $c=s=4$.

\begin{figure}
  \begin{minipage}[t]{0.5\textwidth}
  \centering
  \includegraphics[width=\textwidth]{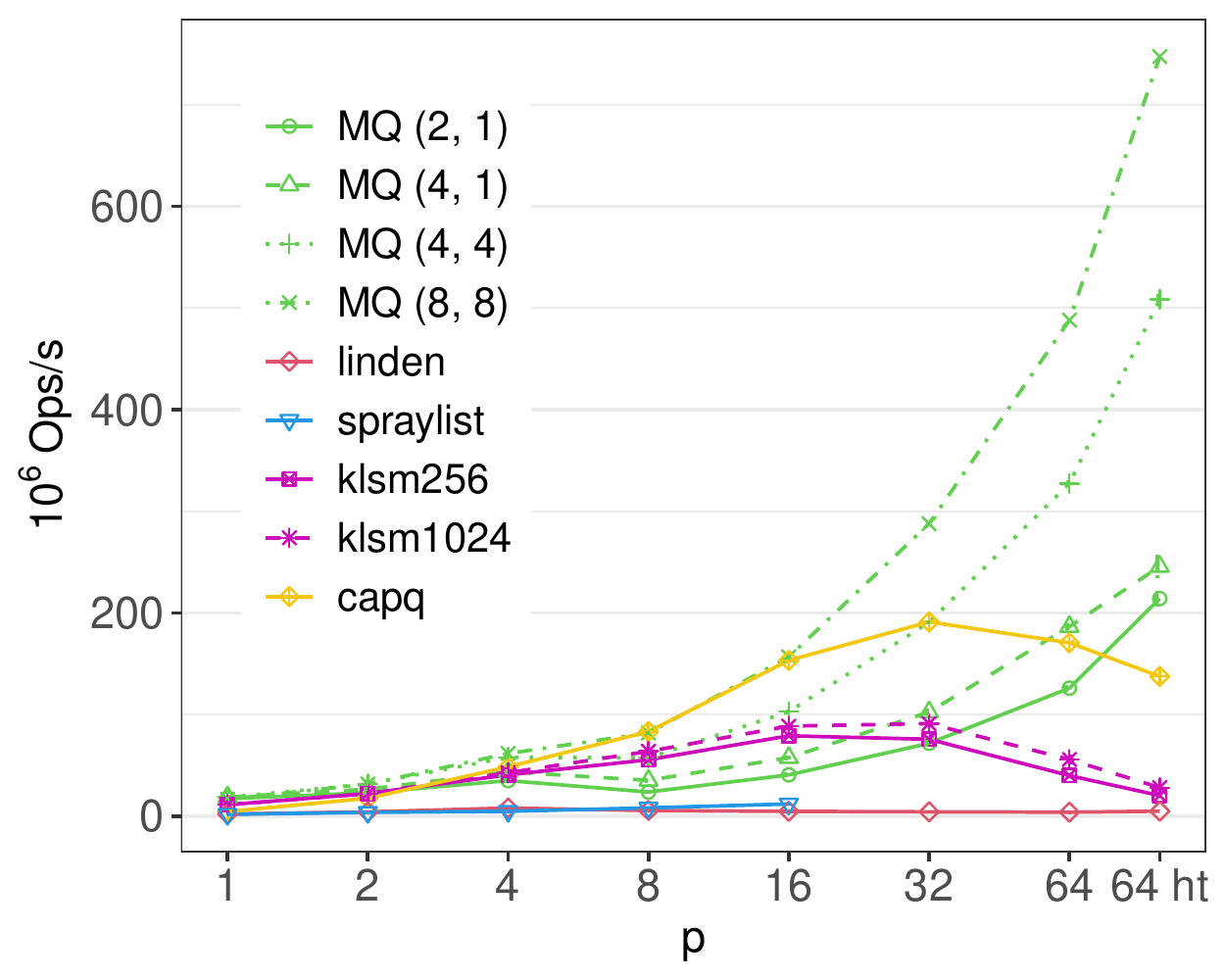}
\end{minipage}
  \begin{minipage}[t]{0.5\textwidth}
  \centering
  \includegraphics[width=\textwidth]{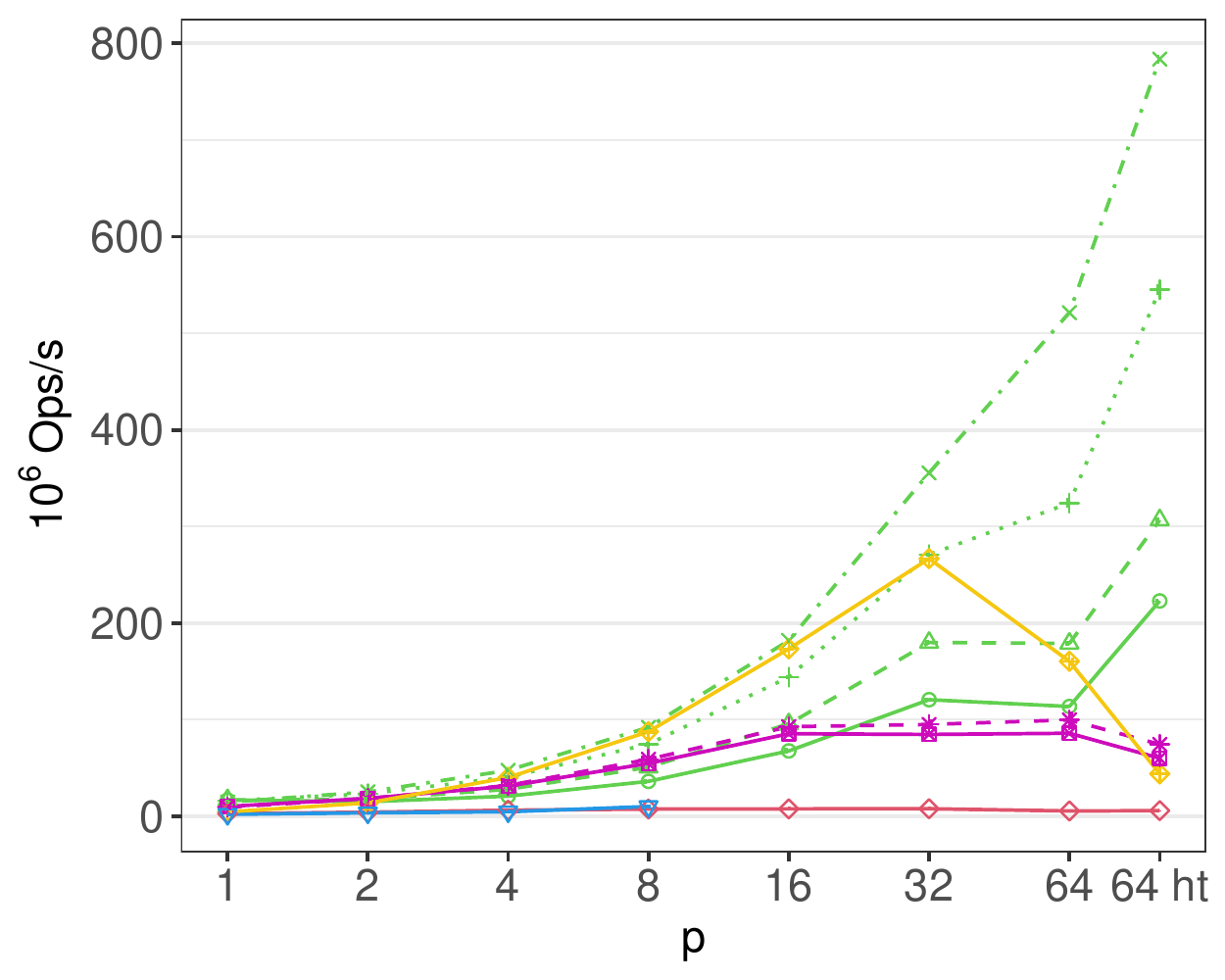}
\end{minipage}
  \caption{\label{fig:throughput_compare}Throughput comparison between the MultiQueue and other implementations. On the left is machine $A$, and on the right machine $B$.}
\end{figure}
As seen in Figure~\ref{fig:throughput_compare}, the CAPQ is among the fastest
priority priority queues for up to $32$ threads.  Beyond that, it accesses the
centralized queue so often that contention deteriorates performance.  This
could probably be remedied with different settings of the tuning parameters but
only with a commensurate effect on even higher rank errors and delays.  The
very loose MultiQueue with $c=s=8$ exhibits very good scalability.  Even the
higher-quality variants can take advantage of all hardware threads and
eventually outperform the CAPQ.  The $k$-LSM scales up to around $16$ threads
while Linden and Spraylist scale very poorly even on very few threads. The
authors of the $k$-LSM show that the scalability of the $k$-LSM can be improved
by selecting higher values for $k$ such as $4096$ at the cost of lower quality
\cite{gruber2016benchmarking}.  Machines $A$ and $B$ have similar performance
with $B$ having a slight advantage when all threads are used but suffering from
the switch from one to two sockets.  Unfortunately, the Spraylist crashed in
our setup on both machines with higher thread count, so we had to exclude it
from the SSSP benchmark below.

\begin{figure}
  \begin{minipage}[t]{0.5\textwidth}
  \centering
  \includegraphics[width=\textwidth]{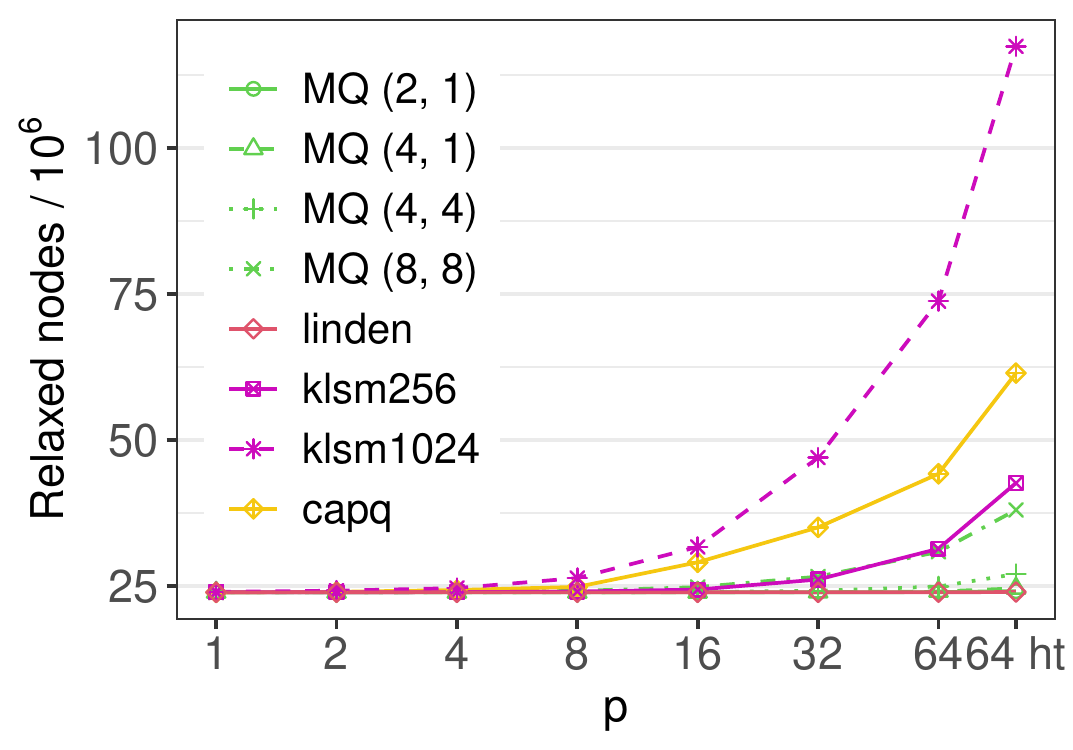}
\end{minipage}
  \begin{minipage}[t]{0.5\textwidth}
  \centering
  \includegraphics[width=\textwidth]{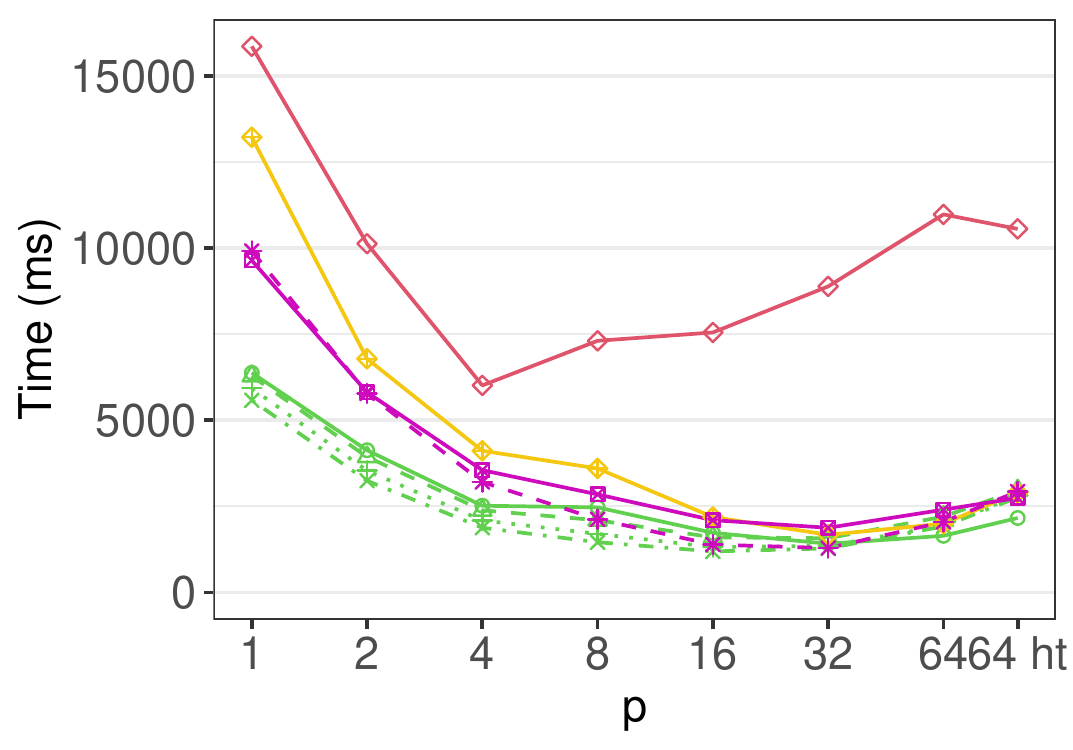}
\end{minipage}
\caption{\label{fig:sssp_usa} The SSSP benchmark on the \emph{USA} graph running on machine $A$.}
\end{figure}
\begin{figure}
  \begin{minipage}[t]{0.5\textwidth}
  \centering
  \includegraphics[width=\textwidth]{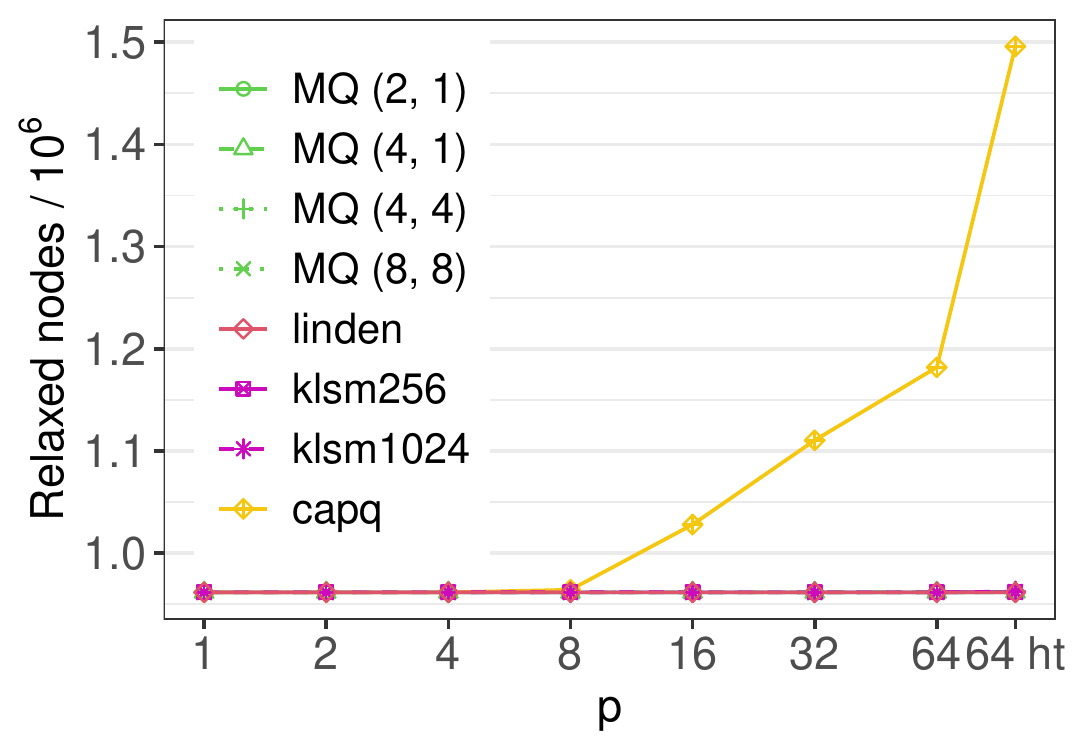}
\end{minipage}
  \begin{minipage}[t]{0.5\textwidth}
  \centering
  \includegraphics[width=\textwidth]{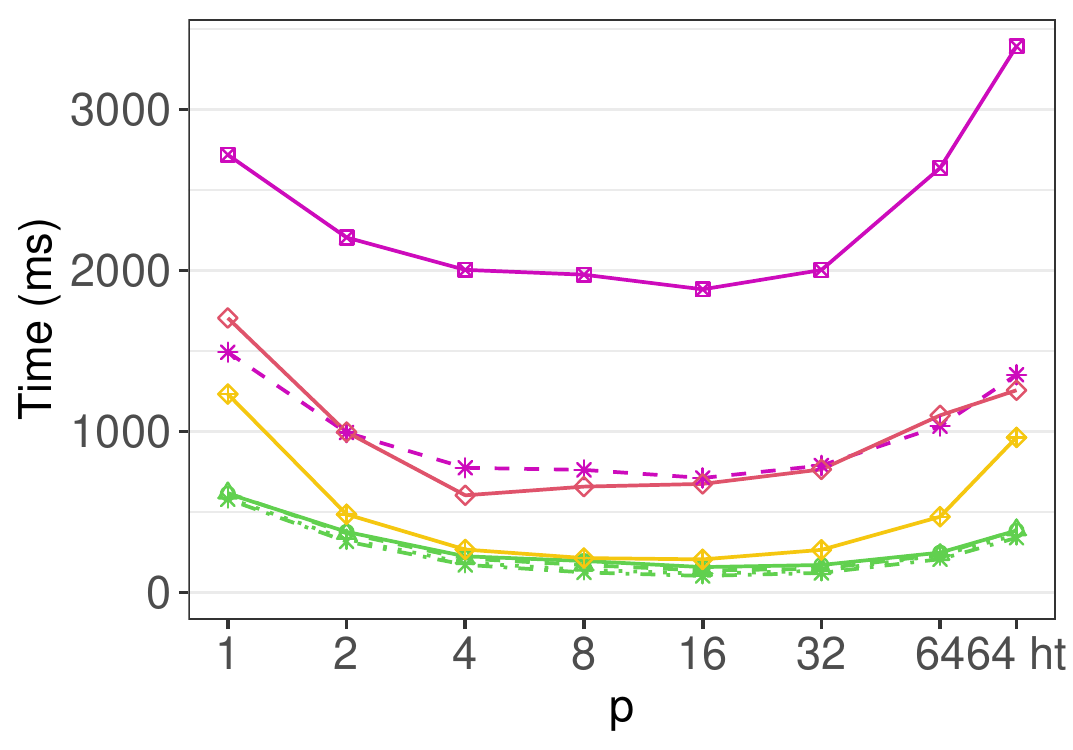}
\end{minipage}
\caption{\label{fig:sssp_rhg_20}The SSSP benchmark on a random hyperbolic graph with $2^{20}$ nodes, an average degree of $16$ and $\gamma = 2.3$ running on machine $A$.}
\end{figure}

The parallel SSSP benchmark uses Dijkstra's algorithm to calculate the shortest
paths from one node to all other nodes in a weighted graph. Dijkstra's
algorithm has to be adapted to be used in a parallel setting with relaxed
priority queues. Sagonas and Winblad  use a similar benchmark for the CAPQ and
describe the algorithm in greater detail \cite{sagonas2016contention}. The
algorithm terminates as soon as the queue is empty.  To detect when the queue
is empty, they use the property of the CAPQ and $k$-LSM that if \Id{deleteMin}
fails for all threads some time after the last insertion happened, the queue
must be empty. This does not hold for MultiQueues, thus we cannot rely solely
on \Id{deleteMin} to guarantee correct termination.  Instead, we use a dedicated
emptiness detection routine, where a thread checks $c$ designated local queues
for emptiness if it thinks that the queue is empty. The MultiQueue is empty if
all threads have successfully completed this emptiness check some time after
the last insertion has happened.
We report the time to solve the SSSP problem for different thread counts as
well as the number of nodes that were extracted from the queue and then
relaxed. We used real road networks and artificial
random hyperbolic graphs (rhg) as benchmark instances.  The road network
\emph{USA}\footnote{\url{http://users.diag.uniroma1.it/challenge9/download.shtml}}
has about $24$m nodes and $58$m edges.  The road network
\emph{GER}\footnote{\url{https://i11www.iti.kit.edu/resources/roadgraphs.php}}
has about $20$m nodes and $42$m edges.  The weights on these graphs represent
the travel time. We used a modified version of the KaGen
framework~\cite{funke2017communication} to generate random hyperbolic graphs
with $2^{20}$ and $2^{22}$ nodes and the geometric distances as edge weights.
Figure~\ref{fig:sssp_usa} and~\ref{fig:sssp_rhg_20} give
exemplary results for the \emph{USA} road network and the rhg with $2^{20}$
nodes on machine $A$.
On the road map graphs, all the relaxed PQs considerably outperform the Linden
queue.  However, none of them scales particularly well beyond 16 threads with
the MultiQueue variants performing best.  Despite having considerably better
rank errors than the CAPQ, the klsm1024 leads the algorithm to processes many
more nodes on the \emph{USA} graph. The $k$-LSM variants have very high delays,
which indicates that the delay is a distinctive metric to measure the quality
of relaxed priority queues.  Rhgs paint a different picture: While only the
CAPQ leads to nodes being processed more than once, both $k$-LSM variants are
noticeably slower than both the CAPQ and MultiQueues.  MultiQueues lead to very
low overhead considering the extracted nodes in all our benchmarks and are very
competitive at the same time.  Especially in algorithms where processing the
individual elements is expensive, this could be a decisive factor.

\section{Conclusions and Future Work}\label{concl}

MultiQueues are a simple and efficient approach to relaxed concurrent
priority queues.  They allow a transparent trade-off between throughput
and quality and considerably outperform previous approaches in at
least one of these aspects.  An important open problem is to complete
the theoretical analysis to encompass stickiness.  We believe that
further practical improvements could be possible by better avoiding
conflicting queue accesses of the sticky variant even when few queues
are used.  It would also be interesting to make MultiQueues contention
aware to achieve higher quality and more graceful degradation than the
simple binary switch between local and global access used in the
CAPQ~\cite{sagonas2016contention}.

A conceptual contribution of our paper is in introducing the quality
measure of \emph{delay} as a complement to \emph{rank error}.   We give evidence
that this is important for the actual performance of applications such
as shortest path search and explain how to measure it efficiently.
 We are
currently having a closer look at using CPQs for branch-and-bound
where it seems that the parallel performance of the applications is
provably tied to the delays incurred by the CPQ.
Beyond priority queues, it would be interesting to see whether our
approach to \textquote{wait-free} locking can be used for other
applications, e.g., for FIFOs.

\clearpage
\bibliography{diss,references}

\clearpage
\appendix

\section{More Experiments}

\subsection{SSSP Benchmark}\label{ss:sssp}

\begin{figure}[h]
  \begin{minipage}[t]{0.5\textwidth}
    \centering
    \includegraphics[width=\textwidth]{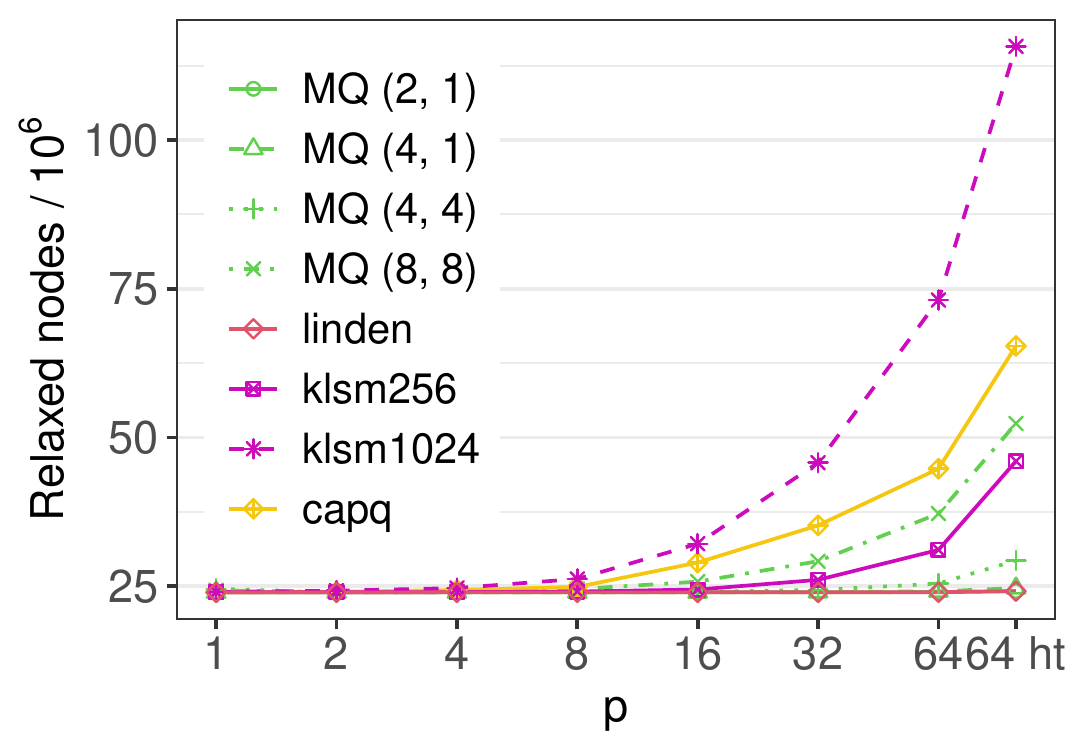}
  \end{minipage}
  \begin{minipage}[t]{0.5\textwidth}
    \centering
    \includegraphics[width=\textwidth]{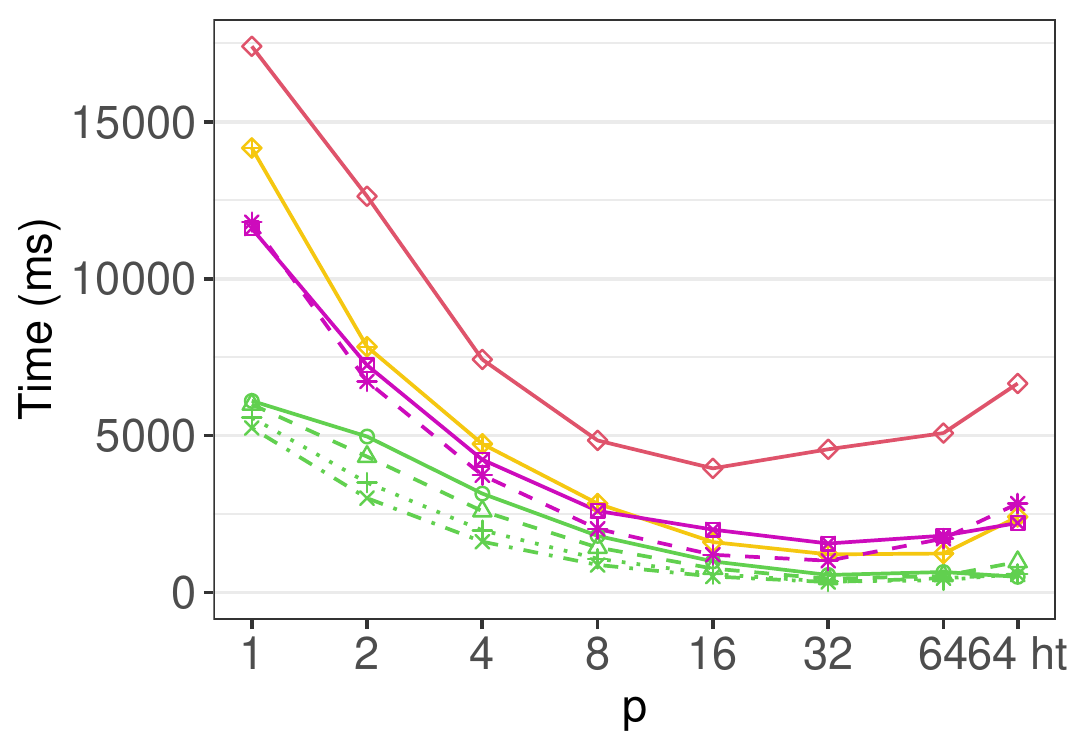}
  \end{minipage}
  \caption{\label{fig:sssp_usa_icx} The SSSP benchmark on the \emph{USA} graph running on machine $B$.}
\end{figure}
\begin{figure}[h]
  \begin{minipage}[t]{0.5\textwidth}
  \centering
  \includegraphics[width=\textwidth]{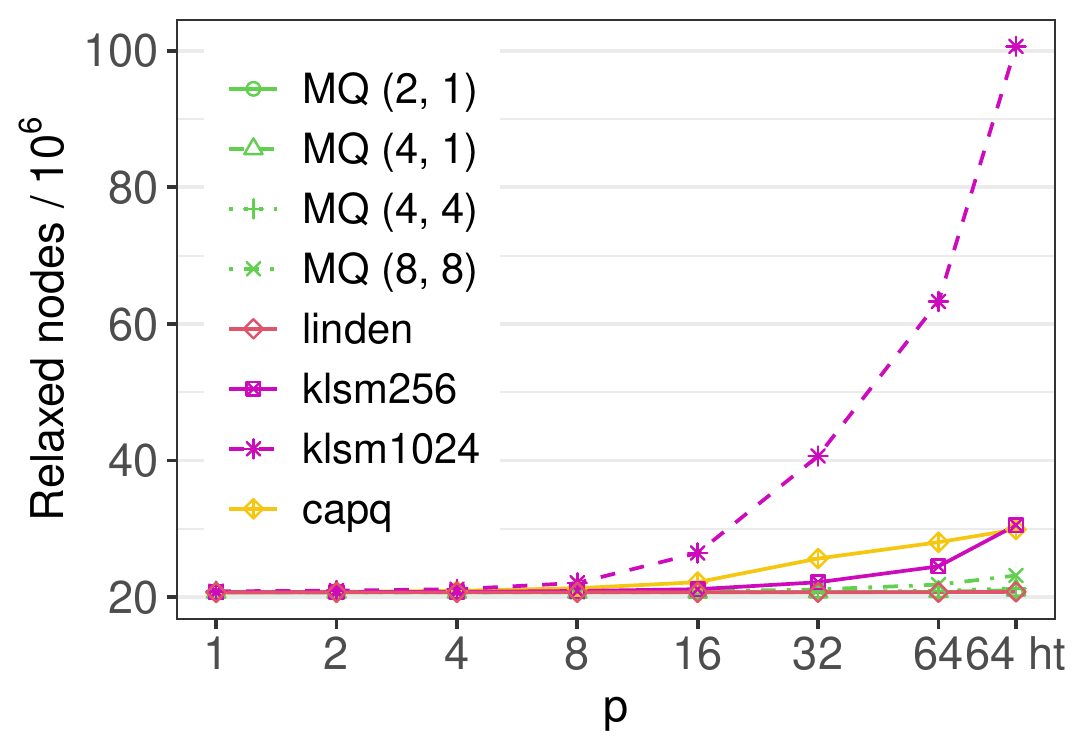}
\end{minipage}
  \begin{minipage}[t]{0.5\textwidth}
  \centering
  \includegraphics[width=\textwidth]{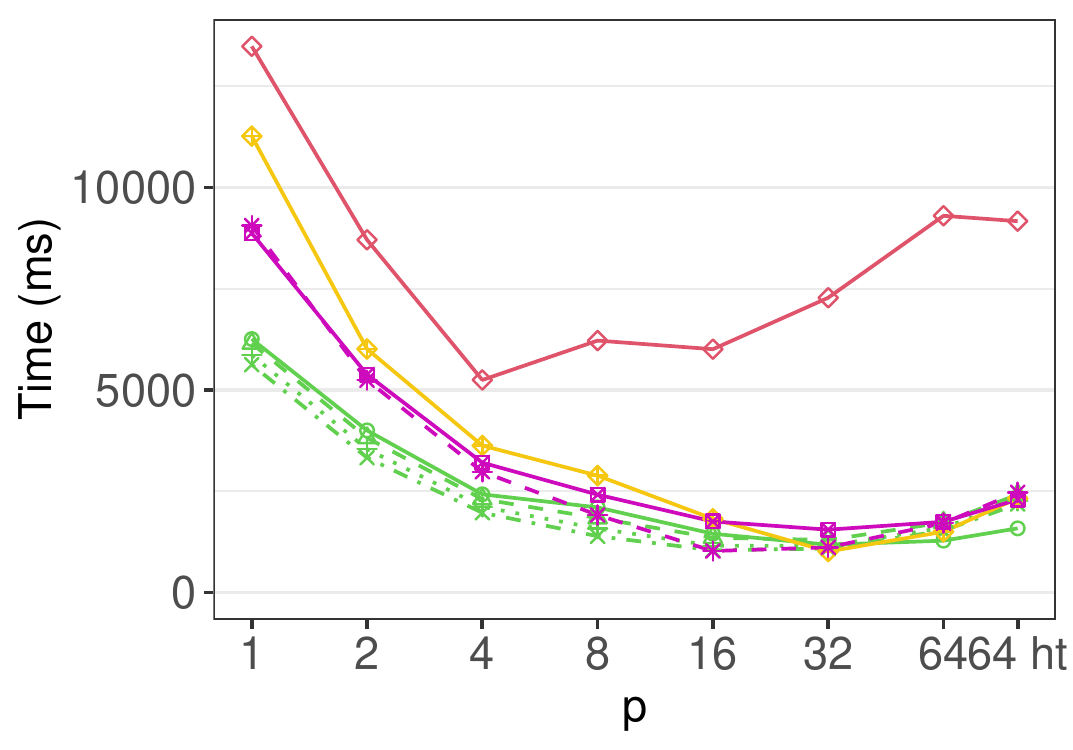}
\end{minipage}
\caption{\label{fig:sssp_ger} The SSSP benchmark on the \emph{GER} graph running on machine $A$.}
\end{figure}

\begin{figure}[h]
  \begin{minipage}[t]{0.5\textwidth}
    \centering
    \includegraphics[width=\textwidth]{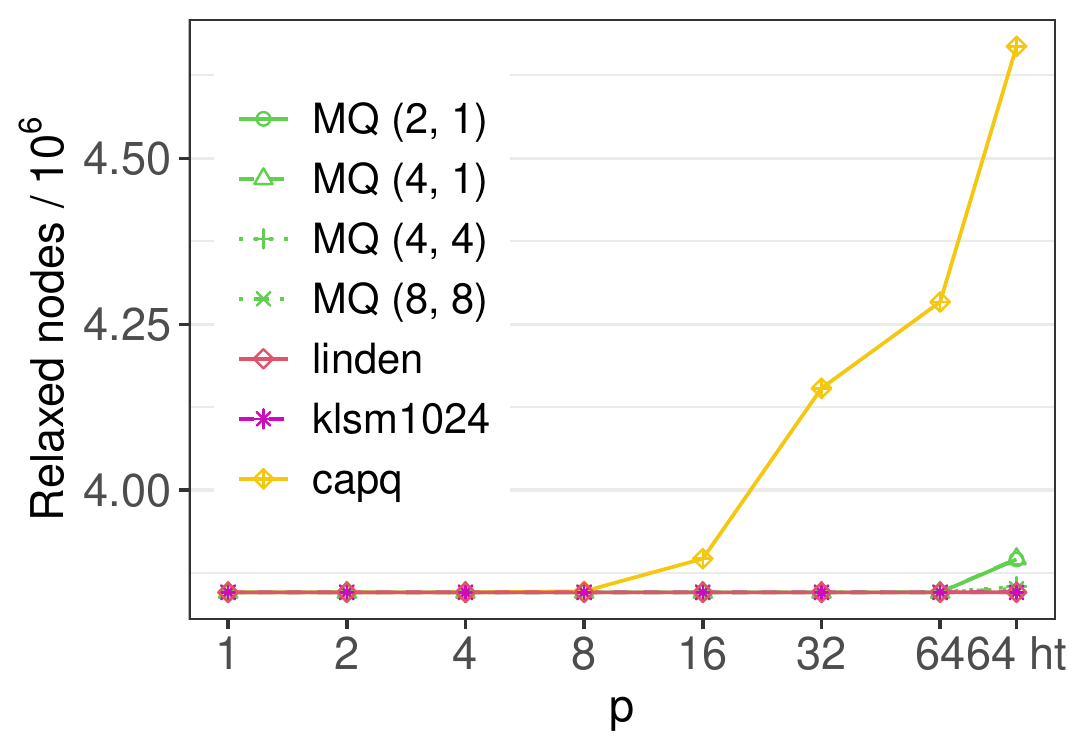}
  \end{minipage}
  \begin{minipage}[t]{0.5\textwidth}
    \centering
    \includegraphics[width=\textwidth]{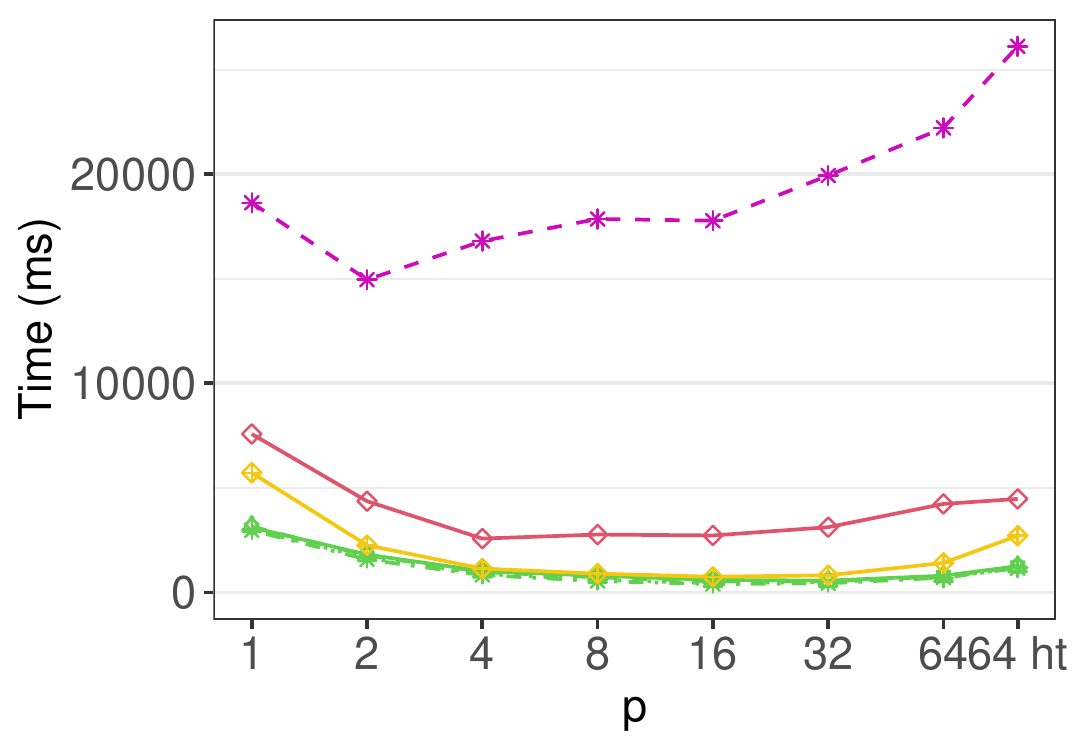}
  \end{minipage}
\caption{The SSSP benchmark on a random hyperbolic graph with $2^{22}$ nodes, an average degree of $16$ and $\gamma = 2.3$ running on machine $A$.}
\end{figure}

\end{document}